\documentclass[12pt,leqno]{article}
\usepackage{latexsym}
\newtheorem{theorem}{Theorem}
\newtheorem{property}{Property}
\newtheorem{corollary}{Corollary}
\newtheorem{lemma}{Lemma}

\newtheorem{definition}{Definition}

\newcommand{\blackslug}{\mbox{\hskip 1pt \vrule width 4pt height 8pt 
depth 1.5pt \hskip 1pt}}
\newcommand{\qed}{\quad\blackslug\lower 8.5pt\null\par\noindent}
\newenvironment{proof}{\par\noindent{\bf Proof:}}{\qed \par}

\newcommand{\cH}{\mbox{${\cal H}$}}

\newcommand{\cR}{\mbox{${\cal R}$}}

\title{Similarity-Projection structures: 
the logical geometry of Quantum Physics 
\thanks{This work was partially supported 
by the Jean and Helene Alfassa fund for 
research in Artificial Intelligence}
}

\author{Daniel Lehmann\\School of Engineering 
and Center for the Study of Rationality  
\\Hebrew University, \\Jerusalem 91904, Israel 
}
\date{May 2008}

\begin{document}
\maketitle
\begin{abstract}
Similarity-Projection structures abstract the numerical properties of 
real scalar product of rays and projections in Hilbert spaces 
to provide a more general framework for Quantum Physics. They are 
characterized by properties that possess direct physical meaning.
They provide a formal framework that subsumes both classical boolean logic 
concerned with sets and subsets and quantum logic concerned 
with Hilbert space, closed subspaces and projections.
They shed light on the role of the phase factors that are central
to Quantum Physics. The generalization of the notion of a self-adjoint 
operator to SP-structures provides a novel notion that is free of linear
algebra.
\end{abstract}

\section{Introduction} \label{sec:intro}
In~\cite{Whitney:matroids}, H. Whitney abstracted the properties of linear
dependence from the setting of vector spaces. This paper represents a 
similar endeavor to abstract the properties of both linear dependence and
projections on closed subspaces from the vector space structure 
of Hilbert spaces.

A system of Quantum Physics is described by a set $\Omega$ of pure states.
In traditional presentations those pure states are modeled as rays, i.e., 
one-dimensional subspaces, of a Hilbert space. 
The main structure possessed by $\Omega$ is its real scalar product.
Given any two rays $x$, $y$ their real scalar product $p(x, y)$ is a real
number in the interval $[0 , 1]$, customarily described as a {\em transition
probability}. 
This quantity is physically meaningful and can be measured in experiments. 
It seems to be the only physically meaningful quantity: 
the only physical property that can be directly measured. 
The purpose of this paper is to study the properties of this quantity. 

This paper studies the properties of the real scalar product of rays in Hilbert
spaces. Surprisingly, such a study has not been pursued very actively so far. 
An algebraic characterization of the properties of real scalar product of
rays in Hilbert space would be interesting, but is not the primary
goal we are seeking. Some of those properties are not satisfied by the
spaces in which Quantum Physics is done, which include superselection rules.
For example the following is a property of $p$, the real scalar product of
rays that is satisfied in all Hilbert spaces but not when superselection
rules are introduced: for any distinct rays $x$, $y$, there exists a ray
$z$ such that \mbox{$0 < p(x, z) < 1$}. This paper's goal is to propose
a list, as extensive as possible, of properties of the real scalar product
of rays {\em that are physically meaningful and satisfied in all spaces
used by Quantum Physics, including classical systems and superselection rules}.

Phase factors play a central role in the thinking of Quantum physicists.
We shall examine the nature of those phase factors, ask whether they can
be defined in terms of the real scalar product of rays.
We shall see that the cosine of those phase factors are definable in terms
of the real scalar product of rays. We shall consider whether the phase factors
themselves have physical meaning
or whether only some trigonometric function of those phase factors is 
meaningful.

In a previous paper Lehmann, Engesser and Gabbay in~\cite{LEG:Malg}
proposed a qualitative study of projections in Hilbert spaces and proposed
M-algebras as an abstraction of properties of projections in Hilbert
spaces meaningful for Quantum Physics. The present paper builds on this
first effort and shares its philosophy.
This paper is a direct successor of~\cite{Qsuperp:IJTP} which
is a concrete study of Hilbert spaces, but failed to give a proper analysis
of the phase factors and of~\cite{Lehmann_andthen:JLC} which proposes a numberless
analysis of projections on subspaces.
 
\section{Similarity} \label{sec:sim-function}
The Similarity-Projection structures (from now on, SP-structures) 
will be introduced gently and slowly.
We are defining structures that include a non-empty set (the carrier) 
$\Omega$. The elements of $\Omega$ are to be thought of as pure states. 
Elements of $\Omega$ will indeed be called states.

A characteristics of Quantum Physics is that pure states
have a dual aspect: they are both states and questions (i.e, observables).
A state \mbox{$x \in \Omega$} can be understood both as a state as in 
``the system is in state $x$'' and as a question like ``let us measure whether 
the system is in state $x$ or not''. Given two states $x$ and $y$, 
if a system in state $x$ is asked whether it is in $y$, there is, in
Quantum Physics, a certain ``probability'' that the answer will be positive.
Given two pure states $s_{1}$ and $s_{2}$, one can measure the probability 
that one will obtain $s_{2}$ when measuring, in state $s_{1}$, 
whether $s_{2}$ holds or not. 
Think, for example, about the simplest 
of quantum systems: a particle of spin $1/2$, let $s_{1}$ be the state
$\mid + \rangle$ in which the spin is up in the $z$-direction, and $s_{2}$
be the state in which the spin is up in the $x$-direction. 
We know that the measurement of the spin in the $x$-direction will, 
on a system in state $s_{1}$ give the answer {\em up} with probability $1/2$ 
and the answer {\em down} with the same probability. 

The first structural ingredient in the definition of SP-structures is therefore
a real function \mbox{$p: \Omega \times \Omega \longrightarrow \cR$}. 
If $x$ and $y$ are states, the real number
$p(x, y)$ is to be understood as the {\em similarity} of $x$ to $y$, or, in
the language used by physicists, the {\em transition probability} between $x$
and $y$.

\section{Hilbert and classical SP-structures} \label{sec:HC-SP-structures}
We shall now present two paradigmatical examples 
of such similarity functions $p$. 
The first example covers what we shall call Hilbert SP-structures. 
Assume \cH\ is a Hilbert space and $\Omega$ is the set of unit vectors
of \cH. For any \mbox{$\vec{x}, \vec{y} \in \Omega$} define 
\mbox{$p(\vec{x}, \vec{y})$} to
be the real scalar product of $\vec{x}$ and $\vec{y}$:
\mbox{$p(\vec{x}, \vec{y}) =$} 
\mbox{$\mid \langle \vec{x} , \vec{y} \rangle \mid^{2}$}.
Note that we depart from the presentation that dates back at least to von 
Neumann of taking $\Omega$ to be the set of rays, i.e., 
one-dimensional subspaces of \cH. We consider unit vectors, not rays. This
is, in fact, closer to the every day practice of physicists.

The second example consists of an arbitrary set $\Omega$ and a similarity 
function
defined by: \mbox{$p(x, y) = 1$} if \mbox{$x = y$} and \mbox{$p(x, y) = 0$} 
otherwise. We shall call such structures {\em classical} SP-structures.

\section{Symmetry and Non-negativity} \label{sec:first-ass}
We shall now list a number of properties of the similarity $p$ that we want
to assume in any SP-structure. We shall draw some consequences 
of those assumptions as we proceed.

Since we are dealing with structures of the type 
\mbox{$\langle \Omega , p \rangle$}
it is natural to define as equivalent any two elements of $\Omega$ that behave in
exactly the same way as far as $p$ is concerned.

\begin{definition} \label{def:equivalence}
Any two states \mbox{$x, y \in \Omega$} are said to be {\em equivalent}, and
we write \mbox{$x \sim y$} iff for any \mbox{$z \in \Omega$}, one has:
\mbox{$p(x, z) = p(y, z)$}.
\end{definition}
The relation $\sim$ is obviously an equivalence relation.
In classical SP-structures, one has \mbox{$x \sim y$} iff
\mbox{$x = y$}. In Hilbert SP-structures two unit vectors are equivalent
iff they differ by a phase factor.

\subsection{Symmetry} \label{sec:symmetry}
Our first assumption is a symmetry assumption.
\begin{property}[Symmetry] \label{prop:symmetry}
For any \mbox{$x , y \in \Omega$}, \mbox{$p(y, x) = p(x, y)$}. 
\end{property}
Symmetry is an experimentally verifiable and fundamental
property of Quantum Mechanics, see, 
e.g., the Law of Reciprocity in~\cite{Peres:QuantumTheory}, p. 35.
It is satisfied by scalar product of rays. It is also obviously satisfied
in classical SP-structures.

Nevertheless Symmetry may be telling us more
about our intellectual processes, our logic, than about the structure
of the physical world out there. If we accept the idea that states possess 
the dual aspects of {\em states the world is in} and of 
{\em states we can test for} and that two states are linked by the fundamental
\mbox{$p(x, y)$}, rejecting Symmetry would be akin to rejecting the idea that
those dual aspects of states are aspects of the same entity, and imply we are
dealing with two different types of entities. 
%Symmetry is a good reason to be weary of the interpretation of $p$ as a
%probability. The natural interpretation of $p$ as a probability is to
%interpret \mbox{$p(x, y)$} as the conditional probability \mbox{$p(y \mid x)$}
%of passing the test $y$ once one has passed the test $x$. But symmetry is not
%a property of conditional probabilities and requiring it seems inconsistent 
%with probability theory.

\subsection{Non-negativity} \label{sec:non-negativity}
Our second assumption is that the similarity $p$ is nonnegative: 
\begin{property}[Non-negativity] \label{prop:non-negativity}
For any \mbox{$x , y \in \Omega$}, \mbox{$p(x, y) \geq 0$}. 
\end{property}
This is requested by the interpretation of $p$ as
a ``probability'' and is obviously satisfied by the scalar product of rays
and in classical SP-structures. It seems that Non-negativity does not
tell us anything about the physical world but is a logical
requirement following from the way our experiments are built.

Since $0$ has a special meaning, as the smallest possible value for $p$, it
is natural to pay special attention to those pairs $x, y$ for which
\mbox{$p(x, y) = 0$}. Following common usage we shall say that $x$ and $y$
are orthogonal and write \mbox{$x \perp y$} iff \mbox{$p(x, y) = 0$}.
Note that \mbox{$y \perp x$} iff \mbox{$x \perp y$}, by Symmetry. 
Similarly we shall say
that $x$ is orthogonal to a set $A$ of states and write \mbox{$x \perp A$} iff
\mbox{$p(x, A) = 0$}. Note that for any state $x$, \mbox{$x \perp \emptyset$}.
We shall use the notation \mbox{$B \perp A$} to mean: for every
\mbox{$x \in A, y \in B$}, one has: \mbox{$x \perp y$}.

\begin{definition} \label{def:ortho-set}
A set $A$ of states will be called an {\em ortho-set} iff
any two distinct elements of $A$ are orthogonal: for any \mbox{$x, y \in A$}
such that \mbox{$x \neq y$}, one has \mbox{$x \perp y$}.
\end{definition}
Note that the empty set is an ortho-set 
and so is any singleton set.
Ortho-sets play a central role in our analysis. They represent states
that correspond to different values of an observable physical quantity.

We shall now generalize $p$ to accept not a single state, 
but any ortho-set of states as a second argument. 
If \mbox{$x \in \Omega$} and \mbox{$A \subseteq \Omega$} is an ortho-set, 
we define
\[
p(x, A) \: = \: \sum_{y \in A} p(x, y).
\]
The (finite or infinite) sum above is independent of the order of summation.
Note that $p(x, A)$ is either a nonnegative real number or $+ \infty$.

\section{Boundedness, subspaces} \label{sec:boundedness}
We may now introduce our next requirement. 
\begin{property}[Boundedness] \label{prop:boundedness}
For any state \mbox{$x \in \Omega$} and any ortho-set $A$, 
\mbox{$p(x, A) \leq 1$}.
\end{property}
The last inequality should be understood as: 
$p(x, A)$ is finite and at most one. 
Again this is a fundamental property in Quantum Physics. The elements of 
an ortho-set $A$ represent different possible answers to a unique
test. The sum of the ``probabilities'' of obtaining certain answers cannot be
greater than one. Boundedness is satisfied both in Hilbert and in classical
SP-structures. Again, Boundedness seems to be a logical requirement, following
from our interpretation of orthogonal states as corresponding to different
values and of similarity as a transition probability.

\begin{definition} \label{def:subspace}
If $A$ is an ortho-set, the subspace \mbox{$\bar{A} \subseteq \Omega$} 
generated by $A$ is defined by:
\mbox{$\bar{A} \: = \:$}
\mbox{$\{x \in \Omega \mid p(x, A) = 1 \}$}. 
The ortho-set $A$ is said to be a basis for $\bar{A}$. 
A {\em basis} is a basis for $\Omega$.
A subspace is a set of states \mbox{$X \subseteq \Omega$} 
such that there exists
some ortho-set $A$ such that \mbox{$y = \bar{A}$}.
\end{definition}
In classical structures \mbox{$\bar{A} = A$}.

In the following lemma, and throughout this paper we shall assume that 
the structure \mbox{$\langle \Omega , p \rangle$} satisfies 
all the assumptions previously made. 
In Lemma~\ref{le:01}, therefore, $p$ is assumed to satisfy 
Symmetry, Non-negativity and Boundedness.

\begin{lemma}\label{le:01}
Let $A$ be an ortho-set.
For any \mbox{$x \in \Omega$}, \mbox{$p(x, A) \in [0 , 1]$}.
In particular, \mbox{$p(x, y) \in [0 , 1]$} for any \mbox{$y \in \Omega$}.
\end{lemma}
\begin{proof}
By Non-negativity, we have \mbox{$p(x, A) \geq$} $0$.
By Boundedness, \mbox{$p(x, A) \leq 1$}.
The singleton \mbox{$\{ y \}$}
is an ortho-set and therefore \mbox{$p(x, y) =$}
\mbox{$p(x, \{ y \}) \in [0 , 1]$}.  
\end{proof}

Any state orthogonal to each of the states of an ortho-set $A$ is orthogonal
to every state in the subspace generated by $A$.
\begin{lemma} \label{le:O-S}
Suppose \mbox{$x \in \Omega$} is a state and \mbox{$A \subseteq \Omega$} 
is an ortho-set such that \mbox{$x \perp A$}. Then,
\mbox{$x \perp \bar{A}$}.
\end{lemma}
\begin{proof}
Since $A$ is an ortho-set and we have \mbox{$p(x, A) =$} $0$, 
the set \mbox{$A \cup \{x\}$} is an ortho-set.
By Boundedness and Symmetry then we have: for any \mbox{$y \in \bar{A}$}
\mbox{$p(y, A) + p(y, x) \leq$} $1$.
But \mbox{$p(y, A) =$} $1$ and therefore \mbox{$p(y, x) =$} $0$.
\end{proof}

\section{O-Projection and consequences} \label{sec:o-projection_cons}
\subsection{O-projection} \label{sec:o-projection}
The next property we want to consider deals with orthogonal projections.
\begin{property}[O-Projection] \label{prop:Projection}
Suppose \mbox{$x \in \Omega$} is a state and \mbox{$A \subseteq \Omega$} 
is an ortho-set such that \mbox{$p(x, A) < 1$}. 
Then there exists a state \mbox{$y \in \Omega$} with the following properties:
\begin{enumerate}
\item \label{perpA} \mbox{$y \perp A$}, i.e., \mbox{$p(y, A) = 0$}, i.e.,
\mbox{$A \cup \{ y \}$} is an ortho-set, and
\item \label{pxAy} \mbox{$p(x, A) + p(x, y) = 1$}.
\end{enumerate}
\end{property}
O-Projection should remind the reader of the Gram-Schmidt process. 
Physically, the ortho-set $A$ represents certain values of a given observable
and therefore can be interpreted as a test: is the state $x$ in $A$ or not.
If \mbox{$p(x, A) <$} $1$ the answer to the question above may, with a certain 
``probability'' be ``no''. If the answer is indeed ``no'' the system is left
in a state $y$ that satisfies the three conditions above. 
The scalar product 
can be seen to satisfy those conditions, when $y$ is the projection of $x$ 
on the subspace $A^{\perp}$ orthogonal to $A$. In a classical system,
 \mbox{$p(x, A) <$} $1$ implies \mbox{$p(x, A) =$} $0$ and we can take
\mbox{$y =$} $x$.
The conditions of O-Projection seem to be logical requirements.

\begin{lemma} \label{le:reflex}
For any states \mbox{$x , y \in \Omega$}, if \mbox{$x \sim y$} then
\mbox{$p(x, y) =$} $1$. In particular, \mbox{$p(x, x) =$} $1$. 
\end{lemma}
\begin{proof}
\begin{sloppypar}
Since $\{ y \}$ is an ortho-set, if it were the case that
\mbox{$p(x, y) <$} $1$, there would exist, by O-Projection, 
some state $z$ such that
\mbox{$p(y, z) =$} $0$ and \mbox{$p(x, y) + p(x, z) =$} $1$.
But \mbox{$p(x, z) =$} \mbox{$p(y, z) =$} $0$ and we conclude that
\mbox{$p(x, y) =$} $1$.
\end{sloppypar}
\end{proof}

\subsection{Bases: existence and size} \label{sec:bases}
\begin{lemma} \label{le:basis_o}
Let $A$ be some ortho-set and assume that
\mbox{$B \subseteq \bar{A}$} is such that, for every \mbox{$x \in \bar{A}$},
\mbox{$p(x, B) = 1$}, then $B$ is a basis for $\bar{A}$.
\end{lemma}
\begin{proof}
We only need to show that $B$ is an ortho-set.
Let \mbox{$x, y \in B$}, \mbox{$x \neq y$}.
We have \mbox{$p(x, B) = 1$}. But
\mbox{$1 = p(x, B) \geq p(x, x) + p(x, y)$} by Non-negativity. 
But, by Lemma~\ref{le:reflex}, \mbox{$p(x, x) = 1$} and we have
\mbox{$p(x, y) = 0$}, by Non-negativity.
\end{proof}

\begin{theorem} \label{the:basis}
Let $A$ be an ortho-set.
Then there is a basis $B$ such that \mbox{$A \subseteq B$}.
\end{theorem}
\begin{proof}
By ordinal induction, we define an ortho-set set 
\mbox{$B_{\alpha} \subseteq \Omega$} for every ordinal $\alpha$. 
We let \mbox{$B_{0} = A$}. For a limit ordinal $\alpha$
we set \mbox{$B_{\alpha} =$} \mbox{$\bigcup_{\beta < \alpha} B_{\beta}$}. 
For any successor ordinal \mbox{$\alpha + 1$}, if $B_{\alpha}$ is a basis we
set \mbox{$B_{\alpha + 1} =$} \mbox{$B_{\alpha}$}, 
and if $B_{\alpha}$ is not a basis,
we consider some state \mbox{$x \in X$}
such that \mbox{$p(x, B_{\alpha}) < 1$} and we
set \mbox{$B_{\alpha + 1} =$} \mbox{$B_{\alpha} \cup \{ y \}$}, where
\mbox{$y \in \Omega$} is one of the states the existence of which is 
guaranteed by O-Projection.
Clearly we have a chain of ortho-sets and 
there is some ordinal $\beta$ for which \mbox{$B_{\beta + 1} =$} 
\mbox{$B_{\beta}$}.
The set $B_{\alpha}$ is a basis.
\end{proof}

It is a striking property of Hilbert spaces that any two bases have the 
same cardinality. The same holds in SP-structures.
\begin{theorem} \label{the:matroid}
Let $A$, $B$ be orthosets such that \mbox{$B \subseteq \bar{A}$} 
and assume $A$ is finite. 
Then $B$ is finite and \mbox{$\mid B \mid \leq \mid A \mid$}.
An SP-structure that admits a finite basis, will be called 
{\em finite-dimensional} and its dimension is the (common) size of its bases.
\end{theorem}
\begin{proof}
We have \mbox{$\sum_{b \in B} \sum_{a \in A} p(a, b) =$}
\mbox{$\sum_{b \in B} 1 =$} \mbox{$\mid B \mid$}.
But \mbox{$\sum_{a \in A} \sum_{b \in B} p(a, b) \leq$}
\mbox{$\sum_{a \in A} 1 =$} \mbox{$\mid A \mid$}.
We conclude that \mbox{$\mid B \mid \leq \mid A \mid$}.
\end{proof}

We may now define a natural operation on subspaces: orthogonal complement.
\begin{theorem} \label{the:comp}
Let $X$ be any subspace. The set \mbox{$X^{\perp} \: = \:$}
\mbox{$\{ x \in \Omega \mid x \perp X \}$} is a subspace and 
\mbox{$X \: = \:$} \mbox{$(X^{\perp})^{\perp}$}.
\end{theorem}
\begin{proof}
Let $A$ be a basis for $X$. Complete $A$ to a basis \mbox{$A \cup B$},
with $B$ an ortho-set orthogonal to $A$. We shall show that $B$ is a basis for
$X^{\perp}$. First, \mbox{$B \subseteq X^{\perp}$} by Lemma~\ref{le:O-S}.
But, for any state $x$ of $X^{\perp}$, \mbox{$p(x , A) + p(x , B) =$} $1$ and
\mbox{$p(x , A) =$} $0$. One sees that $A$ is a basis for $(X^{\perp})^\perp$.
\end{proof}

\subsection{Projections on subspaces} \label{sec:proj}
The following defines projections on subspaces.
\begin{lemma} \label{le:proj}
If $x$ is a state and $A$ is an ortho-set, such that 
\mbox{$p(x, A) > 0$}
there is a state $y$, such that:
\begin{enumerate}
\item \mbox{$y \in \bar{A}$}, and
\item \mbox{$p(x, y) =$} \mbox{$p(x, A)$}.
\end{enumerate}
\end{lemma}
\begin{proof}
Let $B$ be a basis such that \mbox{$B =$} \mbox{$A \cup C$} with 
\mbox{$C \perp A$}. The existence of such a basis follows from 
Theorem~\ref{the:basis}. The set $C$ is an ortho-set and
\mbox{$p(x, A) + p(x, C) = 1$}.
Since \mbox{$p(x, A) > 0$}, we have \mbox{$p(x, C) < 1$}, and, by O-Projection
there is a state $y$ such that  
\mbox{$y \perp C$}, \mbox{$p(x, C) + p(x, y) = 1$}.
But \mbox{$p(y, A) + p(y, C) = 1$} and \mbox{$p(y, C) = 0$}.
Therefore \mbox{$p(y, A) = 1$}. 
Also \mbox{$p(x, y) =$} \mbox{$1 - p(x, C)$} \mbox{$p(x, A)$}.
\end{proof}

The reader may wonder about the case \mbox{$p(x, A) = 0$}. 
In this case, by Lemma~\ref{le:O-S}, every state $y$ such that 
\mbox{$p(y, A) = 1$} satisfies the condition required, i.e.,
\mbox{$p(x, y) = 0$}.

\section{Factorization and Consequences} \label{sec:factorization}
Our next defining property for SP-structure is a factorization property.

\begin{property}[Factorization] \label{prop:Factorization}
Let $A$ be an ortho-set and $x$ an arbitrary state.
If \mbox{$y, z \in \bar{A}$} and \mbox{$p(x, y) = p(x, A)$}, then
\mbox{$p(x, z) = p(x, y) \, p(y, z)$}.
\end{property}
Factorization implies that $p(x, A)$ is the maximum of all $p(x, y)$ for
\mbox{$y \in \bar{A}$} and that every such $p(x, y)$ can be factored out 
through the state taking this maximum.
Factorization has been described in 
Theorem~1 of~\cite{Qsuperp:IJTP}.
The meaning of Factorization, for Physics, is that, if one knows that
in state $y$ some observable $A$ has a specific value,
then the probability of a transition from $x$ to $y$ is the product
of the probability of measuring this specific value (in $x$) 
times the transition probability from the state
obtained after the measurement to $y$.
Factorization seems to be a logical requirement relating tests 
to two propositions one of which entails the other: if $A$ entails $B$, testing
for $A$ may be done by testing first for $B$ and then for $A$.

\begin{theorem} \label{the:max}
\begin{sloppypar}
For any state \mbox{$x \in \Omega$} and any ortho-set $A$, \mbox{$p(x, A) =$}
\mbox{$\max(\{ p(x, y) \mid y \in \bar{A} \}$}.
Therefore if $B$ is an ortho-set such that
\mbox{$\bar{B} =$} \mbox{$\bar{A}$}, one has 
\mbox{$p(x, A) =$} \mbox{$p(x, B)$}.
From now on, if $X$ is a subspace we shall allow ourselves the use
of the notation \mbox{$p(x, X)$}.
Also, if $X$ and $Y$ are subspaces such that \mbox{$X \subseteq Y$}, then,
for any \mbox{$x \in \Omega$}, one has
\mbox{$p(x , X) = p(x , Y) \, p(t(x , Y) , X)$}.
\end{sloppypar}
\end{theorem}
In the last equation note that in the case $t(x , Y)$ is not defined, we have
\mbox{$p(x , Y) = 0$} and therefore we consider the product 
on the right hand side
of the last equation to be defined and equal to zero.
\begin{proof}
By Lemma~\ref{le:proj} there is some \mbox{$z \in \bar{A}$} such that
\mbox{$p(x, A) =$} \mbox{$p(x, z)$} and by Factorization we have, for every
\mbox{$y \in \bar{A}$}, \mbox{$p(x, y) =$}
\mbox{$p(x, z) \, p(z, y) \leq$}
\mbox{$p(x, z)$}.
The remainder follows easily.
\end{proof}

The Factorization property has many consequences that will be presented now.
The first one concerns the relation of equivalence between states.

\subsection{Similarity and Equivalence} \label{sec:simandequiv}
\begin{lemma} \label{le:rho_alpha_zero}
If \mbox{$a \perp b$}, \mbox{$x \perp a$} and
\mbox{$p(y, a) + p(y, b) = 1$}, then we have \mbox{$p(x, y) =$} 
\mbox{$p(x, b) \, p(y, b)$}.
\end{lemma}
\begin{proof}
Let \mbox{$A = \{ a , b \}$}. The set $A$ is an ortho-set.
If \mbox{$p(x, b) = 0$}, \mbox{$p(x, A) = 0$} and, by Lemma~\ref{le:O-S},
\mbox{$p(x, y) = 0$} and the claim is proved.

Assume, then, that \mbox{$p(x, b) > 0$}. Note that \mbox{$p(b, A) = 1$} and
\mbox{$p(x, A) = p(x, b)$}. By Factorization, then
\mbox{$p(x, y) =$} \mbox{$p(x, b) \, p(b, y)$}.
\end{proof}

\begin{theorem} \label{the:equiv}
Any states \mbox{$x, y \in \Omega$} are equivalent, i.e., 
\mbox{$x \sim y$}, iff
\mbox{$p(x, y) = 1$}. 
\end{theorem}
\begin{proof}
If \mbox{$x \sim y$}, \mbox{$p(x, y) = p(x, x)$}, and we conclude the proof
with Lemma~\ref{le:reflex}.

Suppose, now, that \mbox{$p(x, y) = 1$} and that \mbox{$z \in \Omega$}.
We want to show that \mbox{$p(x, z) = p(y, z)$}. Without loss of generality,
we can assume \mbox{$p(z, x) < 1$}.
By O-Projection, there is some state 
\mbox{$x' \perp x$} such that \mbox{$p(z, x) + p(z, x') = 1$}.
By Boundedness we have \mbox{$p(y, x) + p(y, x') \leq 1$} and therefore
\mbox{$p(y, x') = 0$}. The assumptions of Lemma~\ref{le:rho_alpha_zero}
are satisfied for \mbox{$a = x'$}, \mbox{$b = x$}, \mbox{$x = y$} and
\mbox{$y = z$}. We conclude that \mbox{$p(y, z) =$}
\mbox{$p(y, x) \, p(z, x) =$} \mbox{$p(z, x)$}.
\end{proof}

Theorem~\ref{the:equiv} shows that, 
if \mbox{$p(x , y) =$} $1$, then $x$ and $y$
are equivalent, i.e., behave in exactly the same way as far as $p$ is concerned.
No harm can therefore be caused by identifying any two states $x$, $y$ such
that \mbox{$p(x , y) =$} $1$. 
\begin{definition} \label{def:standard}
An SP-structure \mbox{$\langle \Omega , p \rangle$} is said to be {\em standard}
iff for any \mbox{$x , y \in \Omega$}, \mbox{$p(x , y) =$} $1$ implies 
\mbox{$x = y$}. 
\end{definition}
\begin{theorem} \label{the:quotient}
Let \mbox{$\langle \Omega , p \rangle$} be an SP-structure. The quotient
structure \mbox{$\langle \Omega \, / \, \sim , \bar{p} \rangle$} defined
by \mbox{$\bar{p}(\bar{x} , \bar{y}) =$} \mbox{$p(x , y)$} is a standard
SP-structure and the transformation \mbox{$x \hookrightarrow \bar{x}$}
preserves $p$.
\end{theorem}
In the sequel we shall only consider standard SP-structures, even if we forget
to mention the fact. In other words, we assume, from now on, that
\mbox{$p(x , y) =$} $1$ iff \mbox{$x = y)$}.

\subsection{Relativization} \label{sec:relativization}
We shall also strengthen O-Projection and Theorem~\ref{the:basis} by 
relativizing them to a subspace. The relativization of O-Projection 
shows that any subspace of an SP-structure is an SP-structure.
First, we need the following.
\begin{lemma} \label{le:rsub}
Let $x$, $x'$ be states such that \mbox{$p(x , x') <$} $1$. 
Let $y$ be a state orthogonal to $x'$ such that 
\mbox{$p(x , x') + p(x , y) =$} $1$ as guaranteed by O-Projection.
Then, for any ortho-set $A$ such that \mbox{$x \perp A$} and 
\mbox{$x' \perp A$}, we have \mbox{$y \perp A$}.
\end{lemma}
\begin{proof}
Let $z$ be a state in $A$. Since \mbox{$z \perp x'$}, by factorization
we have \mbox{$p(x , z) \: = \:$}
\mbox{$p(x , y) \, p(y , z)$}.
But \mbox{$x \perp z$} and \mbox{$p(x , z) =$} $0$.
But \mbox{$p(x , y) >$} $0$ and therefore 
\mbox{$p(y , z) =$} $0$.
\end{proof}

\begin{theorem} \label{the:relative_OProj}
Suppose \mbox{$X \subseteq \Omega$} is a subspace, \mbox{$x \in X$} is a 
state and \mbox{$A \subseteq X$} is an ortho-set of $X$ such that 
\mbox{$p(x, A) < 1$}. 
Then any state \mbox{$y \in \Omega$} such that
\begin{enumerate}
\item \label{perpA2} \mbox{$y \perp A$}, i.e., \mbox{$p(y, A) = 0$}, i.e.,
\mbox{$A \cup \{ y \}$} is an ortho-set, and
\item \label{pxAy2} \mbox{$p(x, A) + p(x, y) = 1$}
\end{enumerate}
is a member of $X$.
\end{theorem}
\begin{proof}
We have \mbox{$A \perp X^{\perp}$} and \mbox{$x \perp X^{\perp}$} and
therefore, by Lemma~\ref{le:rsub}, \mbox{$y \perp X^{\perp}$} and,
by Theorem~\ref{the:comp}, \mbox{$y \in X$}.
\end{proof}

We may also relativize Theorem~\ref{the:basis}
\begin{theorem} \label{the:basis_rel}
Let $A$ be an ortho-set and $X$ be a subspace such that
\mbox{$A \subseteq X$}.
Then there is a basis $B$ for $X$ such that \mbox{$A \subseteq B$}.
\end{theorem}
The proof follows that of Theorem~\ref{the:basis}, 
using Theorem~\ref{the:relative_OProj} to show that 
\mbox{$B_{\alpha} \subseteq X$}.

\subsection{Intersections and orthogonal sums of subspaces} 
\label{sec:intersec}
In Theorem~\ref{the:comp} we defined orthogonal complements for subspaces.
Once we have established the meaning of $p(x , X)$ for a subspace $X$, as has
been done in Theorem~\ref{the:max},
the following is an obvious corollary of Theorem~\ref{the:comp}.
\begin{corollary} \label{co:comp2}
For any \mbox{$x \in \Omega$} and any subspace $X$, one has
\mbox{$p(x, X) + p(x, X^{\perp}) =$} $1$.
\end{corollary}
The proof is obvious.

We may now define other natural operations on subspaces.
First we define orthogonal sums.
\begin{theorem} \label{the:sum}
Let \mbox{$X \perp Y$} be {\em orthogonal subspaces}. 
The set \mbox{$X \oplus Y$}, defined to be 
\mbox{$\{ x \in \Omega \mid p(x, X) + p(x, Y) = 1 \}$}, is a subspace.
\end{theorem}
\begin{proof}
Let $A$, $B$ be bases for $X$, $Y$ respectively.
The ortho-set \mbox{$A \cup B$} is a basis for \mbox{$X \oplus Y$} by
Theorem~\ref{the:max}.
\end{proof}

We may now define intersections.
\begin{theorem} \label{the:intersection}
If $X$ and $Y$ are subspaces their intersection \mbox{$X \cap Y$} is also
a subspace.
\end{theorem}
\begin{proof}
We shall build an ortho-set \mbox{$A_{\alpha} \subseteq X \cap Y$} for
every  ordinal $\alpha$, by ordinal induction. 
Let \mbox{$A_{0} =$} $\emptyset$. For a limit ordinal $\alpha$
we set \mbox{$A_{\alpha} =$} \mbox{$\bigcup_{\beta < \alpha} A_{\beta}$}. 
For any successor ordinal \mbox{$\alpha + 1$}, if $A_{\alpha}$ is a basis 
for \mbox{$X \cap Y$}, we set \mbox{$A_{\alpha + 1} =$} \mbox{$A_{\alpha}$}, 
and if $A_{\alpha}$ is not a basis for \mbox{$X \cap Y$}
we consider some state \mbox{$x \in X \cap Y$}
such that \mbox{$p(x, A_{\alpha}) < 1$} and we
set \mbox{$A_{\alpha + 1} =$} \mbox{$A_{\alpha} \cup \{ y \}$}, where
\mbox{$y \in \Omega$} is one of the states the existence of which is 
guaranteed by O-Projection. By Theorem~\ref{the:relative_OProj}, 
\mbox{$y \in X \cap Y$}.
Clearly we have a chain of ortho-sets and 
there is some ordinal $\beta$ for which \mbox{$A_{\beta + 1} =$} 
\mbox{$A_{\beta}$}.
The set $A_{\alpha}$ is a basis for \mbox{$X \cap Y$}.
\end{proof}

\subsection{Uniqueness of projections} \label{sec:uniqueness}
We may now strengthen O-Projection and Lemma~\ref{le:proj}.
We then prove a fundamental result on subspaces: 
the projections guaranteed by O-Projection and Lemma~\ref{le:proj} 
are unique and independent of the basis considered.

\begin{theorem} \label{the:proj_one}
In a standard SP-structure, 
if $x$ is a state and $A$ is an ortho-set, such that 
\mbox{$p(x, A) > 0$}, the state \mbox{$y \in \bar{A}$} such that 
\mbox{$p(x, y) =$} \mbox{$p(x, A)$} guaranteed by Lemma~\ref{le:proj} is
unique and depends only on $\bar{A}$, not on $A$.
This unique state will be denoted \mbox{$t(x, \bar{A})$}
or by \mbox{$t(x, A)$}.
\end{theorem}
Note that \mbox{$t(x, A)$} is defined only if \mbox{$p(x, A) > 0$} and that:
\mbox{$p(t(x, A), A) = 1$}, \mbox{$p(x, t(x, A)) = p(x, A)$}
and for any \mbox{$y \in \Omega$} such that \mbox{$p(y, A) = 1$} one has
\mbox{$p(x, y) =$} \mbox{$p(x, t(x, A)) \, p(t(x, A), y)$}.
\begin{proof}
Suppose both $y_{i}$ \mbox{$i = 0 , 1$} satisfy the conditions.
By Factorization we have
\mbox{$p(x, y_{i}) =$} \mbox{$p(x, y_{i+1}) \, p(y_{i+1}, y_{i})$}
for \mbox{$i = 0 , 1$} where \mbox{$1 + 1 = 0$}.
We conclude that \mbox{$p(y_{0}, y_{1}) =$} $1$ and, 
by Theorem~\ref{the:equiv},
that \mbox{$y_{0} \sim y_{1}$}.
Suppose now that $B$ is an ortho-set such that \mbox{$\bar{B} =$}
\mbox{$\bar{A}$}. We have \mbox{$y_{0} \in \bar{B}$} and
\mbox{$p(x, y_{0}) =$} \mbox{$p(x, B)$} by Theorem~\ref{the:max}.
Therefore $y_{0}$ is the projection of $x$ on $B$.
\end{proof}

\subsection{Completion is a closure operation} \label{sec:closure}
The following shows that the completion of an ortho-set into the subspace
it generates has the character of a closure operation.

\begin{theorem} \label{the:subspaces}
Let \mbox{$A , B \subseteq \Omega$} be ortho-sets. The following properties are
equivalent:
\begin{enumerate}
\item \label{sub}
\mbox{$\bar{A} \subseteq \bar{B}$},
\item \label{A}
\mbox{$A \subseteq \bar{B}$},
\item \label{leq}
for any \mbox{$x \in \Omega$} \mbox{$p(x, A) \leq p(x, B)$}.
\end{enumerate}
\end{theorem}
\begin{proof}
Item~\ref{sub} clearly implies item~\ref{A} since \mbox{$A \subseteq \bar{A}$}.
Item~\ref{leq} implies item~\ref{sub} by Boundedness: 
\mbox{$1 =$} \mbox{$p(x, A) \leq$}
\mbox{$p(x, B) \leq$} $1$.
Let us show that item~\ref{A} implies item~\ref{leq}.
Assume \mbox{$A \subseteq \bar{B}$} and \mbox{$x \in \Omega$}.
The ortho-set $A$ can be extended into a basis for the subspace $\bar{B}$
by Theorem~\ref{the:basis}. Therefore \mbox{$p(x, A) \leq$}
\mbox{$p(x, \bar{B}) =$} \mbox{$p(x,B)$} by Theorem~\ref{the:proj_one}.
\end{proof}

\subsection{An iterative description of O-Projection} \label{sec:iter}
We can also strengthen O-Projection.
\begin{theorem} \label{the:o-proj_one}
If $A$ is an ortho-set and $x$ a state such that \mbox{$p(x, A) < 1$}
then there is a unique state $y$ such that
\mbox{$y \perp A$} and \mbox{$p(x, A) + p(x, y) =$} $1$.
This state $y$ is \mbox{$t(x, \bar{A}^{\perp})$} and therefore depends
only on $\bar{A}$ and not on $A$.
\end{theorem}
\begin{proof}
Suppose \mbox{$y_{i} \perp A$} and \mbox{$p(x, A) + p(x, y_{i}) =$} $1$
for \mbox{$i = 0 , 1$}.
We have \mbox{$y_{i} \in A^{\perp}$} and \mbox{$p(x, y_{i}) =$}
\mbox{$p(x, A^{\perp})$} by Lemma~\ref{co:comp2}.
We conclude by Theorem~\ref{the:proj_one}.
\end{proof}

Property~\ref{prop:Projection}, O-Projection claims the existence of the
projection of a state $x$ on the subspace $A^{\perp}$ orthogonal to any 
ortho-set (in fact any subspace) $A$. 
Could we have weakened our assumption and required only the existence of 
such a projection when the ortho-set $A$ is a single state?
The answer is negative: for infinite ortho-sets $A$, 
i.e., for infinite-dimensional subspaces the full force of O-Projection is
needed.
But we shall show now that, for finite-dimensional subspaces, the existence
of o-projections follows from the simple case of a one-dimensional space, with
the help of Property~\ref{prop:Factorization}, Factorization.

\begin{theorem} \label{the:cascade}
Let \mbox{$A = \{a\} \cup A'$} be an ortho-set.
Assume \mbox{$p(x , A) <$} $1$. 
Then, \mbox{$n(x , A) \: = \:$}
\mbox{$n(n(x , A') , \{a\})$}.
\end{theorem}
\begin{sloppypar}
\begin{proof}
Assume \mbox{$p(x , A) <$} $1$. By Non-negativity, 
\mbox{$p(x, A') <$} $1$ and, by O-Projection on $A'$ the state
$n(x , A')$ exists and we have \mbox{$n(x , A') \perp A'$},
\mbox{$p(x , A') + p(x , n(x , A') =$} $1$ and, by Factorization, 
for any state 
\mbox{$y \perp A'$} one has \mbox{$p(x , y) =$}
\mbox{$p(x , n(x , A')) \, p(n(x , A') , y)$}.
We conclude, first, that \mbox{$p(x , a) <$} \mbox{$p(x , n(x , A'))$} and,
by Theorem~\ref{the:equiv}, \mbox{$p(a , n(x , A')) <$} $1$.
Similarly, we see that \mbox{$p(a , n(x , A')) <$} $1$.
Therefore, $b = n(n(x , A') , a)$ is well-defined and we have
\mbox{$b \perp a$}, \mbox{$p(n(x , A') , a) + p(n(x , A') , b) =$} $1$ and
for any state 
\mbox{$y \perp a$} one has \mbox{$p(n(x , A') , y) =$}
\mbox{$p(n(x , A') , b) \, p(b , y)$}.
We notice, first, that \mbox{$p(n(x , A') , b) >$} $0$ since
\mbox{$p(n(x , A') , a) <$} $1$.
For any \mbox{$w \in A'$} we have
\mbox{$p(n(x , A') , w) =$}
\mbox{$p(n(x , A') , b) \, p(b , w)$}.
But \mbox{$n(x , A') \perp w$} and \mbox{$p(n(x , A') , b) >$} $0$
and therefore \mbox{$b \perp w$}.
We have shown that \mbox{$b \perp A$}.

We shall now prove that
\mbox{$p(x , A) + p(x , b) =$} $1$.
We have:
\[
p(x, A) + p(x , b) \: = \:
p(x , a) + p(x , A') + p(x , b) \: = \:
\]
\[
p(x , n(x , A')) \, p(n(x , A') , a) \, + \, p(x , A') \, + \, 
p(x , n(x , A')) \, p(n(x , A') , b) \: = \:
\]
\[
p(x , n(x , A')) \, (p(n(x , A') , a) + p(n(x , A') , b)) \: + \:
p(x , A') \: = \:
\]
\[
p(x , n(x , A')) \: + \: p(x , A') \: = \: 1
\]
\end{proof}
\end{sloppypar}
When the ortho-set $A$ of O-Projection is finite, then the state $y$ can be
described by a sequence of O-projections in which the ortho-set considered 
is a singleton.

\section{Notations} \label{sec:notations}
Before we can express our next requirement we need some notations.
We wish to consider the following general situation.
Suppose \mbox{$X \perp Y$} are orthogonal {\em subspaces} and
let \mbox{$Z = X \oplus Y$} be their orthogonal sum.
Assume $a$ and $b$ are states that are not orthogonal to $Z$, 
i.e., equivalently, none of $a$ or $b$ is orthogonal to both $X$ and $Y$.
The states $a$ and $b$ have projections on $Z$: $t(a, Z)$ 
and $t(b, Z)$.
\begin{theorem} \label{the:comm}
Let \mbox{$X \perp Y$} be subspaces and \mbox{$Z = X \oplus Y$}.
Assume \mbox{$p(a , Z) > 0$}.
If \mbox{$p(a , X) > 0$}, $a$ has a projection on $X$ and
\mbox{$t(a , X) = t(t(a , Z) , X)$}.
\end{theorem}
\begin{proof}
By Theorem~\ref{the:proj_one} since 
\mbox{$p(a , t(a , X)) =$}
\mbox{$p(a , t(a , Z) \, p(t(a , Z) , t(a , X))$} 
\end{proof}

Suppose \mbox{$X \perp Y$}, \mbox{$Z = X \oplus Y$} and \mbox{$a , b \in Z$}.
We expect the quantity \mbox{$p(a , b)$} to be related to 
\mbox{$p(t(a , X) , t(b , X))$} and \mbox{$p(t(a , Y) , t(b , Y))$}.
We shall therefore define two related quantities.
\begin{definition} \label{def:alpha_rho}
Let \mbox{$X \perp Y$}, \mbox{$Z = X \oplus Y$},
\mbox{$a, b \in \Omega$}, \mbox{$p(a , Z) > 0$}, \mbox{$p(b , Z) > 0$},
we shall define:
 \begin{equation} \label{eq:alpha}
\alpha_{X , Y}(a , b) \: = \: p(t(a , Z) , t(b , Z)) \: - \: 
\end{equation}
\[
p(a , X) \, p(b , X) \, p(t(a , X) , t(b , X))
\: - \: p(a , Y) \, p(b , Y) \, p(t(a , Y) , t(b , Y))
\]
\begin{equation} \label{eq:rho}
\rho_{X , Y}(a , b) \: = \:
\end{equation}
\[
2 \, \sqrt{p(a , X) \, p(b , X) \, p(a , Y) \, p(b , Y) \, 
p(t(a, X) , t(b , X)) \, p(t(a , Y) , t(b , Y))}
\]
We use the convention that the product by zero of an undefined quantity 
is equal to zero.
Note, then, that the definitions of $\alpha$ and $\rho$ above are legal 
even if one or more of the
expressions $t(a, X)$, $t(a , Y)$, $t(b , X)$, $t(b , Y)$ is not defined,
or, equivalently if $a$ or $b$ is in $X$ or $Y$ (we assumed $a$ and $b$ are 
not perpendicular to $Z$).
For example, if $t(a , X)$ is not defined, i.e., if \mbox{$a \in Y$}, 
then \mbox{$p(a , X) =$} $0$. 
\end{definition}

Note that \mbox{$\rho_{X , Y}(a , b) =$} $0$ iff $a$ or $b$ is in $X$ or $Y$.
Appendix~\ref{app:Hilbert} describes those quantities in Hilbert spaces.

\begin{lemma} \label{le:alphaaa}
For any \mbox{$a \in Z =$} \mbox{$X \oplus Y$} one has 
\mbox{$\alpha_{X , Y}(a , a) =$} \mbox{$\rho_{X , Y}(a , a)$}.
\end{lemma}
\begin{proof}
\begin{sloppypar}
For any \mbox{$a \in \Omega$}
\mbox{$\rho_{X , Y}(a , a) =$} 
\mbox{$2 \, p(a , X) \, p(a , Y)$} and
\mbox{$\alpha_{X , Y}(a , a) \: = \:$}
\mbox{$1 - p^{2}(a, X) - p^{2}(a , Y)$}.
If \mbox{$a \in Z$} we have \mbox{$p(a , X) + p(a , Y) =$} $1$ and therefore
\mbox{$1 - p^{2}(a, X) - p^{2}(a , Y) \: = \:$}
\mbox{$2 \, p(a , X) \, p(a , Y)$}.
\end{sloppypar}
\end{proof}

\begin{lemma} \label{le:alphminusrho}
Let \mbox{$x \perp y$} be orthogonal states and 
let $X$ be the subspace generated by $x$ and $y$. 
If \mbox{$a , b \in X$} and \mbox{$a \perp b$}, then
\mbox{$\alpha_{x , y}(a , b) \: = \:$} \mbox{$- \, \rho_{x , y}(a , b)$}.
\end{lemma}
\begin{proof}
We have \mbox{$p(a , x) + p(a , y) =$} \mbox{$1 =$}
\mbox{$p(b , x) + p(b , y)$}.
By Boundedness, we have
\mbox{$p(x , a) + p(x , b) \leq$} $1$ and also
\mbox{$p(y , a) + p(y , b) \leq$} $1$. 
We conclude that
\mbox{$p(x , a) + p(x , b) =$} \mbox{$1 =$} 
\mbox{$p(y , a) + p(y , b)$} and therefore 
\mbox{$p(b , y) =$} \mbox{$p(a , x)$} and 
\mbox{$p(a , y) =$} \mbox{$p(b , x)$}.
We see that \mbox{$\rho_{x , y}(a , b) =$} \mbox{$2 \, p(a , x) \, p(b , x)$}
and that \mbox{$\alpha_{x , y}(a , b) =$}
\mbox{$0 - 2 \, p(a , x) \, p(b , x) =$}
\mbox{$- \rho_{x , y}(a , b)$}.
\end{proof}

\begin{definition} \label{def:omega}
Let \mbox{$X \perp Y$} and \mbox{$a, b \in \Omega$}, such that $a$ (resp. $b$) 
is orthogonal to neither $X$ nor $Y$. The projections $t(a, X)$ and $t(a, Y)$ 
(resp. $t(b, X)$ and $t(b, Y)$) are well-defined and 
\mbox{$\rho_{X , Y}(a , b) >$} $0$. 
We shall define 
\[
\omega_{X , Y}(a , b) \: =\:
\frac{\alpha_{X , Y}(a , b)}{\rho_{X , Y}(a , b)}
\]
\end{definition}

\section{A fundamental inequality} \label{sec:inequality}
We present our next assumption.

\begin{property}[Inequality] \label{prop:inequality}
For any orthogonal subspaces \mbox{$X \perp Y$}
and any states \mbox{$x , y \in Z = X \oplus Y$}, 
one has 
\mbox{$\mid \alpha_{X ,Y }(a , b) \mid \: \leq \: \rho_{X , Y}(a , b)$}.
In other words, if \mbox{$\rho_{X , Y}(a , b) >$} $0$, then
\mbox{$- 1 \leq$} \mbox{$\omega_{X , Y}(a , b) \leq$} $1$ and
if \mbox{$\rho_{X , Y}(a , b) =$} $0$, then 
\mbox{$\alpha_{X , Y}(a , b) =$} $0$.
\end{property}

The physical meaning of the property above is not fully understood 
at this point.
It must be related to the two-paths experiments that are so central in 
Quantum Physics.
The deep meaning of Inequality is probably hidden in 
Theorem~\ref{the:continuity}
that shows that it implies a sort of continuity property: if $p(x, y)$
is close to $1$, then $x$ and $y$ are almost equivalent.

In classical SP-structures, Appendix~\ref{app:class} shows 
that for any $X$, $Y$,
$a$, $b$ we have \mbox{$\alpha_{X , Y}(a , b) =$} \mbox{$0 =$} 
\mbox{$\rho_{X , Y}(a , b)$}.
Appendix~\ref{app:Hilbert} studies the Hilbert SP-structures.
In an SP-structure defined by a Hilbert space on the real field, 
\mbox{$\rho_{X , Y}(a , b)$} can be any number in the interval 
\mbox{$[0 , 1 / 2]$} and
\mbox{$\alpha_{X , Y}(a , b)$} is either equal to 
\mbox{$\rho_{X , Y}(a , b)$} or
equal to \mbox{$- \rho_{X , Y}(a , b)$}. But we can say more.
Let us say that $a$ and $b$ are parallel if
\mbox{$\alpha_{X , Y}(a , b) =$} \mbox{$\rho_{X , Y}(a , b) >$} $0$ 
and say that
they are opposite if \mbox{$- \alpha_{X , Y}(a , b) =$} 
\mbox{$\rho_{X , Y}(a , b) >$} $0$. Assume now that $a$ and $c$ are parallel.
If $c$ and $b$ are parallel, then $a$ and $b$ are parallel and if $c$ and $b$ 
are opposite then $a$ and $b$ are opposite.

In an SP-structure defined by a Hilbert space on the complex field, 
one can define an angle 
\mbox{$\varphi_{X , Y}(a , b) \in [0 , 2 \, \pi[$} such that,
for any states $a$, $b$ and $c$ that are neither in $X$ nor in $Y$ one has:
\begin{enumerate}
\item \label{antisym}
\mbox{$\varphi_{X , Y}(b , a) \: = \: - \, \varphi_{X , Y}(a , b)$},
\item \label{potential}
\mbox{$\varphi_{X , Y}(a , b) \: = \: \varphi_{X , Y}(a , c) \: + \: 
\varphi_{X , Y}(c , b)$},
\item \label{equality}
\mbox{$\alpha_{X , Y}(a , b) \: = \: \rho_{X , Y}(a , b) \, 
\cos(\varphi_{X , Y}(a , b))$}.
\end{enumerate}
All angles may appear and a state $a$ is uniquely
characterized by $t(a, X)$, $t(a, Y)$, $p(a, X)$ and an arbitrary angle 
$\psi(a)$ defined
up to an additive constant. Then, \mbox{$\varphi_{X , Y}(a , b) =$} 
\mbox{$\psi(b) - \psi(a)$} (or \mbox{$\psi(a) - \psi(b)$}).

In an SP-structure defined by a Hilbert space on the quaternions 
phases are not angles but unit quaternions.

We see that in each of the four kinds of SP-structures considered 
above some additional
properties hold, but they are different and their physical meaning is unclear.
Since our purpose is to define SP-structures in such a way to cover classical 
and all
three kinds of Hilbert space structures, we do not impose 
any further requirements
on SP-structures.

\section{Consequences of Inequality} \label{sec:cons_ineq}
We shall now examine the consequences of our new
assumption.
First we show that Inequality can be strengthened: the condition 
\mbox{$x \in Z$} is superfluous.
\begin{theorem} \label{the:gen_ineq}
For any states orthogonal subspaces \mbox{$X \perp Y$} and
states \mbox{$a, b$} such that
\mbox{$p(b , X) + p(b , Y) = 1$}, one has
\mbox{$\mid \alpha_{X , Y}(a , b) \mid \leq$}
\mbox{$\rho_{X , Y}(a , b)$}.
\end{theorem}
\begin{proof}
Let \mbox{$Z = X \oplus Y$} be the orthogonal sum of $x$ and $Y$. 
The set $Z$ is a subspace.
If \mbox{$p(a, Z) > 0$}, by Theorem~\ref{the:proj_one}, we have,
for \mbox{$w = t(a , Z) \in Z$},
\mbox{$p(a , b) =$} \mbox{$p(a , w) \, p(w , b)$}, 
\mbox{$p(a , X) =$}
\mbox{$p(a , w) \, p(w , X)$}, and
\mbox{$p(a , Y) =$} \mbox{$p(a , w) \, p(w, Y)$}.
Therefore \mbox{$\alpha_{X , Y}(a , b) =$} 
\mbox{$p(a , w) \, \alpha_{X , Y}(w , b)$}.
Similarly, \mbox{$\rho_{X , Y}(a , b) =$}
\mbox{$p(a , w) \, \rho_{X , Y}(w , b)$}.
By Inequality we have 
\mbox{$\mid \alpha_{X , Y}(w , b) \mid \leq$} 
\mbox{$\rho_{X , Y}(w , b)$}. Since \mbox{$p(a , w) \leq$} $1$ we conclude
that \mbox{$\mid \alpha_{X , Y}(a , b) \mid \leq$}
\mbox{$\rho_{X , Y}(a , b)$}.

If \mbox{$p(a , Z) = 0$}, then, by Lemma~\ref{le:O-S}, \mbox{$p(a , b) = 0$} 
and
\mbox{$\alpha_{X , Y}(a , b) = 0$}.
\end{proof}

Our next result is a continuity property: if $p(x, y)$ is close to one, then,
for any $z$, $p(x, z)$ is close to $p(y, z)$.
\begin{theorem} \label{the:continuity}
For any \mbox{$x, y, z \in \Omega$}, one has:
\[
p(x, z) \leq p(y, z) \: + \: 1 / 2 \, \sqrt{1 - p(x, y)} \: + \: (1 - p(x, y)).
\]
\end{theorem}
\begin{proof}
Assume, for now, that there is some state \mbox{$y' \in \Omega$} such that
\mbox{$y' \perp y$} and \mbox{$p(x, y) + p(x, y') = 1$}. 
Consider any \mbox{$z \in \Omega$}.
By Theorem~\ref{the:gen_ineq}: \mbox{$\mid \alpha_{y, y'}(z, x) \mid \leq$}
\mbox{$\rho_{y, y'}(z, x)$}.
But \mbox{$2 \, \alpha_{y, y'}(z, x) =$}
\mbox{$p(z, x) \: - \: p(z, y) \, p(x, y) \: - \: p(z, y') \, p(x, y') \geq$}
\mbox{$p(z, x) \: - \: p(z, y) \: - \: p(x, y')$}.
Also, \mbox{$\rho_{y, y'}(z, x) \leq \sqrt{p(x, y')}$}.
We conclude that \mbox{$p(z, x) \: - \: p(z, y) \leq$}
\mbox{$1 / 2 \, \sqrt{1 - p(x, y)} \: + \: 1 \: - \: p(x, y)$}.

Now, if \mbox{$p(x, y) < 1$}, there is such a state $y'$ by O-Projection.
Assume, then, that \mbox{$p(x, y) = 1$}. If there is some state 
\mbox{$y' \in \Omega$} orthogonal to $y$, we have \mbox{$1 \geq$}
\mbox{$p(x, y) + p(x, y') =$}
\mbox{$1 + p(x, y')$}, by Boundedness, and we conclude that
\mbox{$p(x, y) + p(x, y') = 1$}.

Let us deal now with the limit case: there is no \mbox{$y' \in \Omega$}
that is orthogonal to $y$. By O-Projection this implies that for every
\mbox{$z \in \Omega$}, \mbox{$p(z, y) = 1$}, and our claim is proved.
\end{proof}
Appendix~\ref{sec:tight} shows that the bounds of Theorem~\ref{the:continuity} 
are tight.

\section{Similarity-preserving mappings} \label{sec:isos}
This section presents preliminary results on mappings that preserve similarity.
We want, now, to consider morphisms between SP-structures.
\begin{definition} \label{def:morph}
A morphism from SP-structure 
\mbox{$S_{1} =$} \mbox{$\langle \Omega_{1} , p_{1} \rangle$}
to SP-structure 
\mbox{$S_{2} =$} \mbox{$\langle \Omega_{2} , p_{2} \rangle$} is 
a function \mbox{$f : \Omega_{1} \longrightarrow \Omega_{2}$} that
preserves similarity, i.e., such that: for any \mbox{$a , b \in \Omega_{1}$}
we have \mbox{$p_{2}(f(a) , f(b)) =$} \mbox{$p_{1}(a , b)$}.
\end{definition}

\begin{theorem} \label{the:inj}
Let $S_{1}$ be a {\em standard} SP-structure.
Any morphism $f$ from $S_{1}$ to any SP-structure $S_{2}$ is injective.
Any such morphism that is surjective is an isomorphism: it has an inverse
that is a morphism.
\end{theorem}
\begin{proof}
If \mbox{$f(a) = f(b)$} then \mbox{$p_{2}(f(a) , f(b)) =$} $1$ and
\mbox{$p_{1}(a , b) =$} $1$. Since $S_{1}$ is standard, we have 
\mbox{$a = b$}.
If $f$ is surjective it is bijective and therefore has an inverse $f^{-1}$.
But \mbox{$p_{1}(f^{-1}(a) , f^{-1}(b)) =$}
\mbox{$p_{2}(f(f^{-1}(a)) , f(f^{-1}(b))) =$}
\mbox{$p_{2}(a , b)$}.
\end{proof}

\begin{theorem} \label{the:iso_basis}
If \mbox{$f : S_{1} \longrightarrow S_{2}$} is an isomorphism,  then
the direct image by $f$ of any basis of $S_{1}$ is a basis of $S_{2}$
and the direct image by $f$ of any subspace of $S_{1}$ is a subspace of
$S_{2}$.
\end{theorem}
The proof is obvious.

We are now interested in studying particular isomorphisms.
First those isomorphisms are automorphisms, i.e., 
isomorphisms from an SP structure to itself. Secondly there is a basis
of invariant states.
\begin{theorem} \label{the:inv_basis}
Let \mbox{$S =$} \mbox{$\langle \Omega , p \rangle$} be
a standard SP-structure and let $B$ be a basis for $S$.  
Let now \mbox{$f : S \longrightarrow S$} be an isomorphism such
that for every element $b$ of the base $B$, we have \mbox{$f(b) = b$}.
Let \mbox{$A \subseteq B$} and let $X$ be the subspace generated by $A$.
We have:
\begin{enumerate}
\item \label{p_inv}
for any \mbox{$x \in \Omega$}, \mbox{$p(x , X) =$} \mbox{$p(f(x) , X)$},
\item \label{g_inv}
$X$ is (globally) invariant under $f$: \mbox{$f(X) =$} $X$,
\item \label{t}
if $x$ is not orthogonal to $X$ then \mbox{$t(f(x), X) =$}
\mbox{$f(t(x , X))$}.
\end{enumerate}
\end{theorem}
\begin{proof}
We have \mbox{$p(x , X) =$} \mbox{$\sum_{a \in A} p(x , a) =$}
\mbox{$\sum_{a \in A} p(f(x), f(a) =$} \mbox{$\sum_{a \in A} p(f(x), a) =$}
\mbox{$p(f(x), X)$}. We have proved item~\ref{p_inv}.
For item~\ref{g_inv}, note that, by Theorem~\ref{the:iso_basis}, $f(X)$ 
is a subspace. By~\ref{p_inv} above, $A$ is a basis for $f(X)$. Therefore
\mbox{$f(X) =$} \mbox{$\bar{A} =$} $X$.
For item~\ref{t}, assume \mbox{$p(x , X) >$} $0$.
Then, by~\ref{p_inv} above, \mbox{$p(f(x), X) > 0$} and both 
\mbox{$t(x , X)$} and
\mbox{$t(f(x) , X)$} are well-defined.
By Theorem~\ref{the:proj_one} \mbox{$t(f(x) , X)$} is the unique state in $X$
such that \mbox{$p(f(x) , t(f(x) , X)) =$} \mbox{$p(f(x) , X)$}.
But \mbox{$p(f(x) , X) =$} \mbox{$p(x , X) =$}
\mbox{$p(x , t(x , X)) =$} \mbox{$p(f(x) , f(t(x , X)))$}.
\end{proof}

\section{Observables} \label{sec:observables}
\subsection{Definition} \label{sec:definition}
We now want to study physical properties. In Classical Physics, those are
numerical values attached to each possible state: a physical property is
represented by a function from the set of all possible states to the real
numbers and essentially (except for some continuity condition 
in the infinite case) any such function represents a possible 
physical property. In Quantum Physics, such physical properties, often
called {\em observables}, are represented by Hermitian, i.e., self-adjoint, 
operators. Such operators have, for Quantum Physics, 
a triple role, that we want to analyze here.
\begin{itemize}
\item
First, as in the classical case, they attach, through their eigenvalues,
values to states. Those values, in the Quantum case, 
are interpreted as mean values. 
\item
But, secondly, those operators define eigensubspaces and projections on those 
that characterize the change of state caused by a measurement,
\item and, thirdly, they represent linear transformations in the state
space, typically interpreted as infinitesimal transformations whose 
commutation properties are significant.
\end{itemize}
In many respects such physical properties (or observables) 
behave in a way that resembles
random variables and therefore maybe they should be termed {\em variables}
instead of observables.

Since the third aspect above is central in Quantum Physics, we shall present
observables as transformations of a special kind on SP-structures.
Observables, viewed as transformations, correspond to Hermitian, 
i.e., self-adjoint bounded linear operators.
\begin{definition} \label{def:observable}
Let \mbox{$\langle \Omega , p \rangle$} be an SP-structure.
An {\em observable} (of the SP-structure) is a transformation
\mbox{$r : \Omega \longrightarrow \Omega$} satisfying the following conditions:
\begin{enumerate}
\item there exists a decomposition of $\Omega$ 
in a denumerable set of non-empty pairwise orthogonal subspaces, i.e., 
there is a denumerable set $I$ 
and non-empty {\em subspaces} $X_{j}$, 
for every \mbox{$j \in I$}
such that:
\begin{itemize}
\item for any \mbox{$j , k \in I$}, if \mbox{$j \neq k$} then
\mbox{$X_{j} \perp X_{k}$} and
\item \mbox{$\Omega \: = \: X_{1} \oplus \ldots \oplus X_{j} \oplus \ldots$}, 
\end{itemize}
\item for every \mbox{$j \in I$}, there is a real number 
$\lambda_{j}$ such that
\begin{itemize}
\item the $\lambda$'s are pairwise different, i.e., 
for any \mbox{$j , k \in I$}, if \mbox{$j \neq k$} then
\mbox{$\lambda_{j} \neq \lambda_{k}$},
\item the $\lambda$'s are bounded, i.e., there some real number $M$ such that
for any $j$, \mbox{$1 \leq j \leq n$}, we have 
\mbox{$\mid \lambda_{j} \mid \: \leq \: M$},
\end{itemize}
\item \label{homo} 
for any \mbox{$a \in \Omega$}, for any $i$, \mbox{$1 \leq i \leq n$}, 
and for any \mbox{$b \in X_{i}$},
we have, if \mbox{$\sum_{j \in I} \lambda^{2}_{j} \, p(a , X_{j}) >$} $0$:
\[
p(r(a) , b) = \frac{\lambda^{2}_{i} \, p(a , b)}
{\sum_{j \in I} \lambda^{2}_{j} \, p(a , X_{j})},
\]
and if \mbox{$\sum_{j \in I} \lambda^{2}_{j} \, p(a , X_{j}) =$} $0$ we have
\mbox{$p(r(a) , b) = p(a , b)$},
\item \label{om}
for any \mbox{$j , k \in I$} such that 
\mbox{$j \neq k$} and for any \mbox{$a , b \in \Omega$} such that 
\mbox{$\omega_{X{j} , X_{k}}(a , b)$} is defined
\begin{enumerate}
\item \label{samesign}
if \mbox{$\lambda_{j} \, \lambda_{k} >$} $0$, then 
\[
\omega_{X_{j} , X_{k}}(r(a) , b) \: = \: \omega_{X_{j} , X_{k}}(a , b)
\]
and, 
\item \label{diffsign}
if \mbox{$\lambda_{j} \, \lambda_{k} <$} $0$, then 
\[
\omega_{X_{j} , X_{k}}(r(a) , b) \: = \: - \, \omega_{X_{j} , X_{k}}(a , b).
\]
\end{enumerate}
\end{enumerate}
\end{definition}
In connection with item~\ref{homo}, note that 
\mbox{$\sum_{j = 1}^{n} \lambda^{2}_{j} \, p(a , X_{j}) =$} $0$ implies that
for any \mbox{$1 \leq i \leq n$} we have 
\mbox{$\lambda^{2}_{i} \, p(a , X_{i}) \: = \: 0$}.
The $\lambda$'s are called eigenvalues. The subspace $X_{j}$ is the 
eigensubspace corresponding to the eigenvalue $\lambda_{i}$ (remember: 
eigenvalues are pairwise different). The assumption that eigenvalues are 
bounded ensures the convergence of the denominator in condition~\ref{homo}.
Note that, if $r$ is an observable with eigenvalues $\lambda_{i}$, it is also
an observable with eigenvalues $c \, \lambda_{i}$ for any $c \neq 0$.
Eigenvalues are defined only up to a (non-zero) multiplicative constant.
Appendix~\ref{app:Herm} shows that Hermitian bounded operators in Hilbert 
spaces define {\em observables}.

In a classical SP-structure, the orthogonal sum is set union and
any state is an eigenvector. An observable is an arbitrary real bounded 
function on states, defined up to a multiplicative non-zero constant.

\subsection{Eigensubspaces} \label{sec:eigensubspaces}
As expected eigensubspaces corresponding to different 
eigenvalues are orthogonal. Projections on eigensubspaces are defined as in 
Theorem~\ref{the:proj_one}.

\begin{theorem} \label{the:eigenspace}
Assume $r$ is an observable with eigensubspaces $X_{j}$ and eigenvalues 
$\lambda_{j}$ for \mbox{$j \in I$}.
\begin{enumerate}
\item \label{s0}
For any \mbox{$x \in \Omega$}, one has
\mbox{$\sum_{j \in I} \lambda^{2}_{j} \, p(x , X_{j}) =$} $0$ iff
there exists some \mbox{$i \in I$} such that \mbox{$x \in X_{i}$} and
\mbox{$\lambda_{i} =$} $0$.
\item \label{perpi}
For any \mbox{$x \in \Omega$} and any \mbox{$i \in I$}, if
\mbox{$x \perp X_{i}$}, then \mbox{$r(x) \perp X_{i}$}.
\item \label{perpj}
For any \mbox{$x \in \Omega$} and any \mbox{$i \in I$} one has
\mbox{$r(x) \perp X_{i}$} iff \mbox{$x \perp X_{i}$} or
\mbox{$\lambda_{i} = 0$} and \mbox{$x \not \in X_{i}$}.
\item \label{proji}
For any \mbox{$x \in \Omega$} and any \mbox{$i \in I$}, if
\mbox{$r(x) \not \perp X_{i}$}, then both \mbox{$t(x , X_{i})$} and
\mbox{$t(r(x) , X_{i})$} are defined and they are equal.
\item \label{fix}
For any \mbox{$i \in I$} and any \mbox{$x \in X_{i}$},
we have \mbox{$r(x) = x$}. 
\item \label{pi}
For any \mbox{$x \in \Omega$} and any \mbox{$i \in I$},
if \mbox{$\sum_{j \in I} \lambda^{2}_{j} \, p(x , X_{j}) >$} $0$, then
\[
p(r(x) , X_{i}) \: = \: 
\frac{\lambda^{2}_{i} \, p(x , X_{i})}
{\sum_{j \in I} \lambda^{2}_{j} \, p(x , X_{j})}
\] 
and if
\mbox{$\sum_{j \in I} \lambda^{2}_{j} \, p(x , X_{j}) =$} $0$, then
\mbox{$p(r(x) , X_{i} =$} \mbox{$p(x , X_{i})$}.
\item \label{fixiff}
For any \mbox{$x \in \Omega$}, \mbox{$r(x) = x$} iff $x$ is an eigenvector,
i.e., a member of some $X_{j}$.
\end{enumerate}
\end{theorem}
\begin{proof}
Let $r$ be an observable. 

For item~\ref{s0}, the {\em if} part is obvious. For the {\em only if} part,
note that, if \mbox{$\sum_{j \in I} \lambda^{2}_{j} \, p(x , X_{j}) =$} $0$,
we must have, for every \mbox{$j \in I$}, 
\mbox{$\lambda^{2}_{i} \, p(x , X_{i}) =$} $0$. 
Therefore \mbox{$p(x , X_{i} >$} $0$ implies \mbox{$\lambda_{i} =$} $0$.
Since all $\lambda$'s are different there is at most one $\lambda_{i}$ equal
to zero and there is some $i \in I$ such that \mbox{$\lambda_{i} =$} $0$ and
\mbox{$p(x , X_{i}) =$} $1$.

In the remainder of this proof we shall use~\ref{homo} 
of Definition~\ref{def:observable} often and without always mentioning it. 

For item~\ref{perpi}, note that, if \mbox{$p(x, a) = 0$} for
\mbox{$a \in X_{i}$}, then, by~\ref{homo} of Definition~\ref{def:observable}, 
\mbox{$p(r(x) , a) = 0$}. Therefore 
\mbox{$x \perp X_{i}$} implies \mbox{$r(x) \perp X_{i}$}.

Consider item~\ref{perpj}, now. 
For the {\em if} part, use item~\ref{perpi} above and notice that if
\mbox{$p(x , X_{i}) < 1$} and \mbox{$\lambda_{i} = 0$} then
\mbox{$\gamma = $}
\mbox{$\sum_{j \in I} \lambda^{2}_{j} \, p(x , X_{j}) >$} $0$
and therefore, for any \mbox{$a \in X_{i}$}
we have \mbox{$p(r(x) , a) =$}
\mbox{$(0 \, p(x , a) \, / \, \gamma =$} $0$.
For the {\em only if} part
assume \mbox{$p(r(x) , X_{i}) = 0$} and \mbox{$p(x , X_{i}) > 0$}.
There is some state \mbox{$a \in X_{i}$} such that \mbox{$p(x , a) > 0$}.
But \mbox{$p(r(x) , a) = 0$} and therefore \mbox{$\lambda_{i} = 0$} and
\mbox{$\sum_{j \in I} \lambda^{2}_{j} \, p(x , X_{j}) >$} $0$, which implies
\mbox{$x \not \in X_{i}$}.

For item~\ref{proji}, assume \mbox{$p(r(x) , X_{i}) > 0$}. 
By item~\ref{perpi} above, 
\mbox{$p(x , X_{i}) > 0$}. Therefore both $t(r(x) , X_{i})$ and
$t(x , X_{i})$ are well-defined. The state $t(r(x) , X_{i})$ is the only
state of $X_{i}$ such that 
\mbox{$p(r(x) , t(r(x) , X_{i}) \leq p(r(x) , a)$} for every state 
\mbox{$a \in X_{i}$}. But 
\mbox{$p(x , t(x , X_{i}) \leq p(x , a)$} for every such state and
\[
p(r(x) , t(x , X_{i}) \: = \:
\frac{\lambda_{i} \, p(x , t(x , X_{i})}{\sum_{j \in I} \lambda^{2}_{j} \, 
p(a , X_{j})} \: \geq \: 
\frac{\lambda_{i} \, p(x , t(x , a)}{\sum_{j \in I} \lambda^{2}_{j} \, 
p(a , X_{j})} \: = \:
p(r(x) , t(x , a)
\]
for any \mbox{$a \in X_{i}$}. We conclude that
\mbox{$t(r(x) , X_{i}) =$} \mbox{$t(x , X_{i}$}.

For item~\ref{fix}, suppose that $x$ is an element of $X_{i}$. 
The state $x$ is orthogonal
to any $X_{j}$ for \mbox{$j \neq i$}. Therefore $r(x)$ is orthogonal to
any such $X_{j}$ by item~\ref{perpi}. We see that \mbox{$r(x) \in X_{i}$}.
By item~\ref{proji} we have
\[
r(x) \: = \: t(r(x) , X_{i}) \: = \: t(x , X_{i}) = x.
\]

Consider item~\ref{pi}, now.
The {\em if} part is item~\ref{fix} above. 
For the {\em only if} part, let us assume that \mbox{$r(x) = x$}. 
By contradiction, assume that there are 
\mbox{$i , j \in I$}, \mbox{$i \neq j$} such that 
\mbox{$x \not \perp X_{i}$} and \mbox{$x \not \perp X_{j}$}.
By item~\ref{s0}, 
\mbox{$\sum_{j \in I} \lambda^{2}_{j} \, p(x , X_{j}) >$} $0$.
By item~\ref{pi}, \mbox{$p(r(x) , X_{i}) \: = \:$} 
\mbox{$\frac{\lambda^{2}_{i} \, p(x , X_{i})}
{\sum_{j \in I} \lambda^{2}_{j} \, p(x , X_{j})}$}
and
\mbox{$p(r(x) , X_{j}) \: = \:$} 
\mbox{$\frac{\lambda^{2}_{j} \, p(x , X_{i})}
{\sum_{j \in I} \lambda^{2}_{j} \, p(x , X_{j})}$}.
But \mbox{$r(x) = x$}, \mbox{$p(x , X_{i}) >$} $0$ and
\mbox{$p(x , X_{j}) >$} $0$ therefore 
\mbox{$\lambda^{2}_{i} \: = \:$}
\mbox{$\lambda^{2}_{j} \: = \:$}
\mbox{$\sum_{j \in I} \lambda^{2}_{j} \, p(x , X_{j})$}.
Since the $\lambda$'s are different, we conclude that 
\mbox{$\lambda_{i} \: = \:$}
\mbox{$- \lambda_{j}$}.
We then see that there can only be one pair of indexes $i , j$ such that
\mbox{$x \not \perp X_{i}$} and \mbox{$x \not \perp X_{j}$}.
We conclude that \mbox{$x \in Z =$} \mbox{$X_{i} \oplus X_{j}$}.
Since \mbox{$\lambda_{i} \, \lambda_{j} <$} $0$, by~\ref{om} of 
Definition~\ref{def:observable}, we have 
\mbox{$\omega_{X_{i} , X_{j}}(x , a) \: = \:$}
\mbox{$- \omega_{X_{i} , X_{j}}(x , a)$} for any \mbox{$a \in \Omega$}
and therefore \mbox{$\omega_{X_{i} , X_{j}}(x , a) \: = \:$} $0$ for any 
state $a$. But, by Lemma~\ref{le:alphaaa}
\mbox{$\omega_{X_{i} , X_{j}}(x , x) \: = \:$} $1$.
\end{proof}

\subsection{Mean values} \label{sec:mean}
The first role of observables: defining mean values, will be discussed now.
\begin{definition} \label{def:mean}
Let $r$ be an observable of the SP-structure \mbox{$\langle \Omega , p \rangle$} 
and let \mbox{$x \in \Omega$} be a state.
The mean value of $r$ in (or at) $x$, denoted $\hat{r}(x)$ 
is defined by:
\mbox{$\hat{r}(x) \: = \:$}
\mbox{$\sum_{i \in I} \lambda_{i} \, p(x, X_{i})$}.
This, possibly infinite, sum is absolutely convergent, since the $\lambda$'s
are bounded. 
\end{definition}
The definition of the mean value resembles an expected value if one interprets
$p$ as a conditional probability. 

\begin{theorem} \label{the:mean_basis}
Assume $r$ is an observable and \mbox{$B_{i}$}, \mbox{$i \in I$} 
is a basis consisting of eigenvectors,
then for any state \mbox{$x \in \Omega$}, we have
\mbox{$\hat{r}(x) \: = \:$}
\mbox{$\sum_{i \in I} \hat{r}(b_{i}) \, p(x , b_{i})$}.
\end{theorem}
\begin{proof}
The states $b_{i}$ that are elements of $X_{j}$ form a basis for $X_{j}$.
By Factorization, then,
\[ 
\sum_{i , b_{i} \in X_{j}} \hat{r}(b_{i}) \, p(x , b_{i}) \: = \:
\sum_{i , b_{i} \in X_{j}} \lambda_{j} \, p(x , b_{i}) \: = \:
\lambda_{j} \, \sum_{i , b_{i} \in X_{j}} p(x , X_{j}) \, p(t(x , X_{j}) , b_{i})
\: = \:
\]
\[
\lambda_{j} \, p(x , X_{j}) \sum_{i , b_{i} \in X_{j}} p(t(x , X_{j}) , b_{i})
\: = \:
\lambda_{j} \, p(x , X_{j})
\]
\end{proof}

Note that Theorem~\ref{the:mean_basis} assumes the basis chosen includes only
eigenvectors of the observable $r$. This condition cannot be dispensed with.

The following shows that mean-value functions are, in a
sense, continuous. Our claim deals only with finite-dimensional structures.

\begin{theorem} \label{the:mean_continuity}
Let $\Omega$ be a finite-dimensional SP-structure.
Let $r$ be an observable.
For any \mbox{$x , y \in \Omega$},
\[ 
\hat{r}(x) - \hat{r}(y) \: \leq \:
(\sum_{i \in I} \lambda_{i}) \, (1 / 2 \, \sqrt{1 - p(x, y)}
\: + \: (1 - p(x, y)))
\]
\end{theorem}
The sum of the eigenvalues is well defined since $I$ is finite, by assumption.
\begin{proof}
\[
\hat{r}(x) \: - \: \hat{r}(y) \: = \:
\sum_{i \in I} \lambda_{i} \, (p(x, X_{i}) \, - \, p(y, X_{i})) \: \leq \:
\]
\[
(\sum_{i \in I} \lambda_{i}) \: ( 1 / 2 \, \sqrt{1 - p(x, y)} \: + \: ( 1 - p(x, y))) 
\]
by Theorem~\ref{the:continuity}.
\end{proof}

\section{Conclusion and future work} \label{sec:future}
This paper has proposed SP-structures as a generalization for the structure
of one-dimensional subspaces in Hilbert spaces. A novel notion of observables
in an SP-structure generalizes self-adjoint, i.e., Hermitian, operators.
A physical system should be viewed as an SP-structure together with a 
set of observables. To be acceptable, the set of observables must be rich 
enough to justify the similarity measure $p$ of the SP-structure.
The value of $p(x , y)$ must be justified by one of the observables at hand.
Defining precisely and studying such sets of observables is probably the next
step.

\section*{Acknowledgements}
I want to thank Jean-Marc L\'{e}vy-Leblond for his continuing help in 
straightening me out on Quantum Physics. I also benefitted from Lev Vaidman's
help: thanks. 
\bibliographystyle{plain}

\begin{thebibliography}{1}

\bibitem{Lehmann_andthen:JLC}
Daniel Lehmann.
\newblock A presentation of quantum logic based on an "and then" connective.
\newblock {\em Journal of Logic and Computation}, 2007.
\newblock to appear, DOI: 10.1093/logcom/exm054.

\bibitem{Qsuperp:IJTP}
Daniel Lehmann.
\newblock Quantic superpositions and the geometry of complex hilbert spaces.
\newblock {\em International Journal of Theoretical Physics}, to appear.

\bibitem{LEG:Malg}
Daniel Lehmann, Kurt Engesser, and Dov~M. Gabbay.
\newblock Algebras of measurements: the logical structure of quantum mechanics.
\newblock {\em International Journal of Theoretical Physics}, 45(4):698--723,
  April 2006.
\newblock DOI 10.1007/s10773-006-9062-y.

\bibitem{Peres:QuantumTheory}
Asher Peres.
\newblock {\em Quantum Theory: Concepts and Methods}.
\newblock Kluwer, Dordrecht, The Netherlands, 1995.

\bibitem{Whitney:matroids}
Hassler Whitney.
\newblock On the abstract properties of linear dependence.
\newblock {\em American Journal of Mathematics}, 57:509--533, 1935.

\end{thebibliography}

\appendix
\section{Classical SP-structures} \label{app:class}
Let $\Omega$ be an arbitrary set and let \mbox{$p(x, y) = 1$} iff
\mbox{$x = y$} and \mbox{$p(x, y) = 0$} otherwise.
We have: \mbox{$x \sim y$} iff \mbox{$x = y$}.
Symmetry is satisfied since equality is symmetric.
Non-negativity is satisfied by definition.
Two states $x$ and $y$ are orthogonal iff they are different.
Any subset $A$ of $\Omega$ is an ortho-set and \mbox{$\bar{A} =$} $A$.
Boundedness is satisfied since \mbox{$p(x, A) = 1$} iff \mbox{$x \in A$}
and \mbox{$p(x, A) = 0$} otherwise.
Suppose, now, that \mbox{$p(x, A) < 1$}. Then, \mbox{$p(x, A) = 0$} and
\mbox{$p(x, A) \, + \, p(x, x) = 1$}. One easily sees that $x$ is a suitable
$y$, showing that $x$ is a suitable $y$
in the definition of the property of O-Projection.
There is only one basis for $\Omega$: $\Omega$ itself.
The property of Factorization is easily established.
Subspaces (i.e., subsets) $X$ and $Y$ are orthogonal iff their
intersection is empty. Orthogonal sum is set union.
If \mbox{$X \perp Y$}, \mbox{$a , b \in X \cup Y$},
\mbox{$\alpha_{X , Y}(a , b) =$} $0$ in all cases
and therefore Inequality is satisfied. 
The quantity \mbox{$\rho_{X , Y}(a , b)$} is also always equal to zero.

\section{Hilbert SP-structures} \label{app:Hilbert}
Assume \cH\ is a Hilbert space on the complex field, 
$\Omega$ is the set of all unit vectors of \cH\ and $p$ is real scalar product:
\mbox{$p(\vec{x}, \vec{y}) = \mid \langle \vec{x} , \vec{y} \rangle \mid^{2}$}.

\subsection{First properties} \label{app:first_prop}
We have: \mbox{$\vec{x} \sim \vec{y}$} iff there is a phase factor
such that \mbox{$\vec{y} = e^{i \varphi} \, \vec{x}$}, by Cauchy-Schwarz.
Symmetry is satisfied since 
\mbox{$\langle \vec{y} , \vec{x} \rangle =$}
\mbox{$\overline{\langle \vec{x} , \vec{y} \rangle}$}.
Non-negativity is satisfied by definition.
Orthogonality has its usual meaning in Hilbert spaces.
An ortho-set is an orthonormal set of vectors. 
Boundedness follows from the existence of an orthonormal basis for \cH\ that
extends any orthonormal set of vectors and the fact that
\mbox{$\vec{x} =$}
\mbox{$\sum_{\vec{b} \in B} \langle \vec{x} , \vec{b} \rangle \, \vec{b}$}. 
Basis has its usual meaning in Hilbert spaces. Subspaces are {\em closed}
subspaces. Let $X$ be any subspace and 
assume that $\vec{x}$ is a unit vector. The quantity \mbox{$p(\vec{x}, X)$}
is the square of the norm of the projection of $\vec{x}$ 
on the subspace $X$.
Let us show that the property of O-Projection is satisfied.
If \mbox{$p(\vec{x}, X) < 1$}, the projection of $\vec{x}$ on the subspace
$X^{\perp}$ orthogonal to $X$, call it $\vec{w}$, is not null.
Let \mbox{$\vec{y} \: = \:$} 
\mbox{$\vec{w} \, / \, \parallel \vec{w} \parallel$}.
Note that \mbox{$\parallel \vec{w} \parallel^{2} =$}
\mbox{$p(x, y)$}. Therefore $1 =$ \mbox{$\parallel \vec{x} \parallel^{2} =$}
\mbox{$p(x , X) + p(x, y)$}.
Then $\vec{y}$ is a unit vector that satisfies the properties
of O-Projection. 
Let us check Factorization.
Assume \mbox{$\vec{x}, \vec{y}, \vec{z}$} are unit vectors and $X$ 
is a subspace. Assume \mbox{$\vec{y}, \vec{z} \in X$} and 
\mbox{$p(\vec{x}, \vec{y}) =$} \mbox{$p(\vec{x}, X)$}.
If \mbox{$\vec{x} \perp X$} then \mbox{$\vec{x} \perp \vec{z}$} and 
Factorization is satisfied.
Otherwise, let $\vec{w}$ be the non-null projection of $\vec{x}$ on $X$
and let \mbox{$\vec{v} =$} 
\mbox{$\vec{w} \: / \: \parallel \vec{w} \parallel$}.
Then \mbox{$\langle \vec{x} , \vec{z} \rangle =$}
\mbox{$\langle \vec{w} , \vec{z} \rangle =$}
\mbox{$\langle \vec{v} , \vec{z} \rangle \, \parallel \vec{w} \parallel$}.
Therefore \mbox{$p(\vec{x}, \vec{z}) =$}
\mbox{$p(\vec{v}, \vec{z}) \, p(\vec{x}, \vec{v})$}.

\subsection{Inequality and phases} \label{app:phases}
Let $X$ and $Y$ be orthogonal subspaces and $\vec{a}$, $\vec{b}$ 
be unit vectors in their orthogonal sum \mbox{$X \oplus Y$}. Let 
$\vec{a}_{X}$, $\vec{b}_{X}$, $\vec{a}_{Y}$ and $\vec{b}_{Y}$ 
be their projections on $X$ and $Y$ respectively.
Assume, first, that \mbox{$\rho_{X , Y}(\vec{a} , \vec{b}) =$} $0$.
Without loss of generality assume that \mbox{$\vec{a}_{Y} =$} $\vec{0}$.
Then we have \mbox{$\vec{a} \in X$}.
If \mbox{$\vec{b}_{X} =$} $\vec{0}$ one easily checks that 
\mbox{$\alpha_{X , Y}(\vec{a} , \vec{b}) =$} $0$.
If \mbox{$\vec{b}_{X} \neq$} $\vec{0}$ one has
\[
\alpha_{X , Y}(\vec{a} , \vec{b}) \: = \:
p(\vec{a} , \vec{b}) \: - \: p(\vec{b}, X) \, 
p(\vec{a}, t(b , X))) \: = \: 0
\]
by Factorization.

Assume, now that \mbox{$\rho_{X , Y}(\vec{a} , \vec{b}) >$} $0$.
By projection the vectors $\vec{a}$ and $\vec{b}$ on $X$ and $Y$ respectively
we have \mbox{$\vec{a} =$} \mbox{$\vec{a}_{X} \: + \: \vec{a}_{Y}$}
and \mbox{$\vec{b} =$} \mbox{$\vec{b}_{X} \: + \: \vec{b}_{Y}$}.
But \mbox{$\parallel \vec{a}_{X} \parallel =$} \mbox{$\sqrt{p(\vec{a} , X)}$}
and similarly for all four terms.
Therefore the unit-vector 
\mbox{$t(\vec{a} , X)$} is \mbox{$\frac{\vec{a}_{X}}{\sqrt{p(\vec{a} , X)}}$} 
and similarly for the other terms.
Therefore 
\[
\langle \vec{a} , \vec{b} \rangle \: = \:
\langle \vec{a}_{X} , \vec{b}_{X} \rangle \: + \:
\langle \vec{a}_{Y} , \vec{b}_{Y} \rangle \: = \:
\]
\[
\sqrt{p(\vec{a} , X) \, p(\vec{b} , X)} \, 
\langle t(\vec{a} , X) , t(\vec{b} , X) \rangle \: + \:
\sqrt{p(\vec{a} , Y) \, p(\vec{b} , Y)} \, 
\langle t(\vec{a} , Y) , t(\vec{b} , Y) \rangle
\]

Let \mbox{$\langle t(\vec{a} , X) , t(\vec{b} , X) \rangle =$}
\mbox{$r \, e^{i \theta}$} and 
\mbox{$\langle t(\vec{a} , Y) , t(\vec{b} , Y) \rangle =$}
\mbox{$s \, e^{i \psi}$}. We have
\[
p(\vec{a} , \vec{b}) \: = \:
\]
\[
\mid \langle \sqrt{p(\vec{a} , X) \, p(\vec{b} , X)} \, t(\vec{a} , X) , t(\vec{b} , X) 
\rangle \: + \:
\langle \sqrt{p(\vec{a} , Y) \, p(\vec{b} , Y)} \, 
t(\vec{a} , Y) , t(\vec{b} , Y) \rangle \mid^{2} \: = \:
\]
\[
p(\vec{a} , X) \, p(\vec{b} , X) \, r^{2} \: + \: 
p(\vec{a} , Y) \, p(\vec{b} , Y) \, s^{2} \: + \:
\]
\[
2 \, \sqrt{p(\vec{a} , X) \, p(\vec{b} , X) \, p(\vec{a} , Y) \, p(\vec{b} , Y)}
\, r \, s \, \cos(\psi - \theta).
\]
But \mbox{$r^{2} =$}
\mbox{$p(t(\vec{a} , X) , t(\vec{b} , X))$} and
\mbox{$s^{2} =$}
\mbox{$p(t(\vec{a} , Y) , t(\vec{b} , Y))$}.
Therefore
\mbox{$\alpha_{X , Y}(\vec{a} , \vec{b}) =$}
\mbox{$\rho_{X , Y}(\vec{a} , \vec{b}) \, \cos(\psi - \theta)$}.
We have proved Inequality.

\section{Tight bounds in Theorem~\ref{the:continuity}} \label{sec:tight}
\begin{sloppypar}
Let $\vec{u}$ and $\vec{v}$ be orthogonal and let:
\mbox{$\vec{x} =$} 
\mbox{$\sqrt{r} \, \vec{u} \: + \: \sqrt{1 - r} \, \vec{v}$} and
\mbox{$\vec{y} =$} 
\mbox{$\sqrt{r - \epsilon} \, \vec{u} \: + \: 
\sqrt{1 - r + \epsilon} \, e^{i \, \delta} \, \vec{v}$}, for 
\mbox{$r \in [0 , 1]$} and $\epsilon$ and $\delta$ close to zero, and
\mbox{$0 < r < 1$}.
\end{sloppypar}

We have
\[
p(\vec{x}, \vec{y}) = r \, (r - \epsilon) \: + \: (1 - r) \,
(1 - r + \epsilon) \: + \: 2 \, \sqrt{r \, (r - \epsilon) \,
(1 - r) \, (1 - r + \epsilon)} \cos(\delta) =
\]
\[
r^{2} \, (1 - \epsilon / r) \: + \: (1 - r)^{2} \, (1 + \epsilon / (1 - r))
\: + \: 2 \, r \, (1 - r) \, \sqrt{(1 - \epsilon / r)(1 + \epsilon / (1 - r))}
\, \cos(\delta) =
\]
\[
r^{2} \, (1 - \epsilon / r) \: + \: (1 - r)^{2} \, (1 + \epsilon / (1 - r))
\: + \: 2 \, r \, (1 - r) \, \sqrt{(1 - \epsilon (1 - 2 r + \epsilon) 
/ (r (1 - r)) }
\, \cos(\delta)
\]
We develop to the second order in $\epsilon$ and $\delta$.
\[
p(\vec{x}, \vec{y}) = r^{2} - r \, \epsilon \: + \: (1 - r)^{2} \: + \:
(1 - r) \epsilon \: + \: 
\]
\[
2 \, r \, (1 - r) (1 - \delta^{2} / 2) \: - \:
(1 - 2r + \epsilon) \, \epsilon \, (1 - \delta^{2} / 2)
\: - \: 1 / 4 \, (1 - 2 r)^{2} / ( r \, (1 - r)) \epsilon^{2} =
\]
\[
r^{2} \: + \:(1 - r)^{2} \: + \: 2 \, r \, (1 - r) \: + \:
\]
\[
- r \, \epsilon \: + \: (1 - r) \epsilon \: - \: (1 - 2 r) \epsilon \: - \:
\]
\[
r \, (1 - r) \, \delta^{2}  \: - \: 
\]
\[
( 1 + 1 / 4 \, (1 - 2 r)^{2} / ( r \, (1 - r)) \epsilon^{2} \: + \: 
O(\epsilon \, \delta^{2})
\]
We conclude that
\[
1 - p(x, y) = r \, (1 - r) \, \delta^{2}  \: + \: 
\epsilon^{2} \, / \, ( r \, (1 - r)) \: + \: 
O(\epsilon \, \delta^{2})
\]
Taking \mbox{$\delta = 0$} and \mbox{$z = \vec{u}$} 
Theorem~\ref{the:continuity}
gives \mbox{$r \leq r - \epsilon + \epsilon + 4 \epsilon^{2}$} showing that
the first term \mbox{$1 / 2 \sqrt{1 - p(x, y)}$} in 
Theorem~\ref{the:continuity} is tight.

\section{Hermitian operators and Observables} \label{app:Herm}
We are interested in showing that in a Hilbert SP-structure, Hermitian 
(i.e., self-adjoint) bounded linear operators are {\em observables} as
defined in Definition~\ref{def:observable}.

Let \cH\ be a separable Hilbert space, $\Omega$ the set of rays of \cH\ 
and let 
$A$ be a Hermitian operator on \cH\, i.e., $A$ is a bounded linear self-adjoint
transformation of \cH. We shall define the transformation 
\mbox{$r : X \longrightarrow X$} in the following way.
Let \mbox{$\vec{x} \in X$}, i.e., $\vec{x}$ is a unit vector of \cH.
If \mbox{$A(\vec{x}) \neq \vec{0}$} we define 
\mbox{$r(\vec{x}) =$} 
\mbox{$A(\vec{x}) \, / \, \parallel A(\vec{x}) \parallel$} and if
\mbox{$A(\vec{x}) = \vec{0}$} we define 
\mbox{$r(\vec{x}) =$} \mbox{$\vec{x}$}.

We are going to show that $r$ is an observable in the sense of 
Definition~\ref{def:observable}.
The set of eigenvalues of $A$ is the set of $\lambda_{i}$'s for
\mbox{$i \in I$}. For \mbox{$i \in I$} the subspace $X_{i}$ is the set
of all unit vectors of the eigensubspace of $A$ corresponding to $\lambda_{i}$.
Let $\vec{a}$, $\vec{b}$ be unit vectors and assume \mbox{$\vec{b}$} is
an eigenvector for $\lambda_{i}$.
If \mbox{$\sum_{j \in I} \lambda^{2}_{j} \, p(a , X_{j}) =$} $0$, then
$\vec{a}$ is an eigenvector for the eigenvalue $0$ and therefore
\mbox{$r(\vec{a}) =$} $\vec{a}$ and \mbox{$p(r(\vec{a}) , \vec{b}) =$}
\mbox{$p(\vec{a} , \vec{b})$} for any unit vector $\vec{b}$.
If \mbox{$\sum_{j \in I} \lambda^{2}_{j} \, p(a , X_{j}) >$} $0$, then
\[
p(r(\vec{a}) , \vec{b}) \: = \: 
\mid \frac{\langle A(\vec{a}) , \vec{b} \rangle}
{\parallel A(\vec{a}) \parallel}
\mid^{2} \: = \:
\mid \frac{\langle \vec{a} , A(\vec{b}) \rangle}
{\parallel A(\vec{a}) \parallel} \mid^{2} \: = \:
\mid \frac{\langle \vec{a} , \lambda_{i} \, \vec{b} \rangle}
{\parallel A(\vec{a}) \parallel}
\mid^{2} \: = \:
\]
\[
\frac{\lambda^{2}_{i} \, p(\vec{a} , \vec{b})}
{\sum_{j \in I} \lambda^{2}_{j} \, p(\vec{a} , X_{j})}.
\]

Let us show, now, that item~\ref{om} of Definition~\ref{def:observable} 
holds. Let \mbox{$j , k \in I$} such that 
\mbox{$j \neq k$} and \mbox{$\vec{a} , \vec{b} \in \Omega$} such that 
neither $\vec{a}$ or $\vec{b}$ are orthogonal to either $X_{j}$ or $X_{k}$.
Let \mbox{$Z = X_{j} \oplus X_{k}$}.
The projections 
\mbox{$t(\vec{a}, Z)$}, \mbox{$t(\vec{b}, Z)$} 
and \mbox{$t(r(\vec{a}), Z)$} exist (by Theorem~\ref{the:eigenspace} 
item~\ref{perpj}).
Let \mbox{$t(\vec{a}, Z) \: = \:$}
\mbox{$\sqrt{r} \, \vec{a}_{1} \: + \: \sqrt{1 - r} \, e^{i \varphi} \, 
\vec{a}_{2}$}
with \mbox{$0 <$}
\mbox{$r =$}
\mbox{$p(t(\vec{a}, Z), \vec{a}_{1} =$}
\mbox{$p(\vec{a}, \vec{a}_{1}) \, / \, 
(p(\vec{a}, \vec{a}_{1}) + p(\vec{a}, \vec{a}_{2})) <$} $1$.
Similarly, let \mbox{$t(b, Z) \: = \:$}
\mbox{$\sqrt{s} \, \vec{b}_{1} \: + \: 
\sqrt{1 - s} \, e^{i \theta} \, \vec{b}_{2}$}.
By the analysis of Appendix~\ref{app:phases}
we have:
\[
\alpha_{X_{j} , X_{k}}(\vec{a} , \vec{b}) \: = \:
\rho_{X_{j} , X_{k}} (\vec{a} , \vec{b}) \, \cos(\theta - \varphi).
\]

\begin{sloppypar}
Assume, first, that both $\lambda_{j}$ and $\lambda_{k}$ are strictly positive.
We have \mbox{$A(\vec{a}) \: = \:$}
\mbox{$\lambda_{j} \, \sqrt{r} \, \vec{a}_{1} \: + \: 
\lambda_{k} \, \sqrt{1 - r} \, e^{i \varphi} \, \vec{a}_{2}$} 
is a non-null vector and its norm is 
\mbox{$R \: = \:$}
\mbox{$\sqrt{\lambda^{2}_{j} \, r \: + \: \lambda^{2}_{k} (1 - r)}$}.
Therefore \mbox{$t(r(a), Z) \: = \:$}
\mbox{$A(\vec{a}) \, / \, R$} and
\begin{equation} \label{eq:same}
\alpha_{X_{j} , X_{k}}(r(\vec{a}) , \vec{b}) \: = \:
\rho_{X_{j} , X_{k}}(r(\vec{a}) , \vec{b}) \, \cos(\theta - \varphi).
\end{equation}
\end{sloppypar}
We see that \mbox{$\omega_{X_{j} , X_{k}}(r(\vec{a}) , \vec{b}) \: = \:$}
\mbox{$\omega_{X_{j} , X_{k}}(\vec{a} , \vec{b})$}.

Assume, now that both $\lambda_{j}$ and $\lambda_{k}$ are strictly negative.
Then, the proper form for a non-null vector of \mbox{$t(r(\vec{a}), Z)$} is now
\mbox{$(- \lambda_{j}) \, \sqrt{r} \, \vec{a}_{1} \: + \:$} 
\mbox{$(- \lambda_{k}) \, \sqrt{1 - r} \, e^{i \varphi} \, \vec{a}_{2}$}.
and we also obtain Equation~\ref{eq:same}.

Assume, now that $\lambda_{j} > 0$ and $\lambda_{k} < 0$.
Then, the proper form for a non-null vector of \mbox{$t(r(\vec{a}), Z)$} is
\mbox{$\lambda_{j} \, \sqrt{r} \, \vec{a}_{1} \: + \:$} 
\mbox{$(- \lambda_{k}) \, \sqrt{1 - r} \, e^{i \varphi + \pi} \, \vec{a}_{2}$}.
Then 
\[
\alpha_{X_{j} , X_{k}}(r(\vec{a}) , \vec{b}) \: = \:
\rho_{X_{j} , X_{k}}(r(\vec{a}) , \vec{b}) \, \cos(\theta - \varphi + \pi).
\]
We obtain
\mbox{$\omega_{X_{j} , X_{k}}(r(\vec{a}) , \vec{b}) \: = \:$}
\mbox{$- \omega_{X_{j} , X_{k}}(\vec{a} , \vec{b})$}.

Similarly, the same obtains if 
\mbox{$\lambda_{j} < 0$} and \mbox{$\lambda_{k} > 0$}.

\end{document}